\documentclass[a4paper,12pt]{article}
\usepackage{amsmath}
\usepackage{graphicx}
\usepackage{enumerate}
\usepackage{natbib}
\usepackage{amssymb}
\usepackage{multirow}
\usepackage[caption = false]{subfig} 
\usepackage{comment}
\usepackage{float}
\usepackage{amsthm}

\newcommand{\braces}{\set}
\newcommand{\brackets}[1]{\left(#1\right)}
\newcommand{\E}[1]{\mathbb{E}\left[#1\right]}
\newcommand{\iid}{\emph{i.i.d.}}

\renewcommand{\P}{\mathbb{P}}
\newcommand{\etas}{ETAS}
\newcommand{\retas}{RETAS}
\newcommand{\rmd}{\mathrm{d}}
\newcommand{\set}[1]{\left\{#1\right\}}

\newtheorem{theorem}{Theorem}[section]

\begin{document}

\title{Spatiotemporal \etas\ model with a
	renewal main-shock arrival process}
\author{Tom Stindl\thanks{
		This research includes computations using the Linux computational cluster Katana supported by the Faculty of Science, UNSW Sydney, and the National Computational Infrastructure (NCI) supported by the Australian Government.}  \,\,and Feng Chen\thanks{Chen was partly supported by a UNSW SFRGP grant.} \\
  Department of Statistics, UNSW Sydney} \date{}

\maketitle

\begin{abstract}
		 This article proposes a spatiotemporal point process model that enhances the classical Epidemic-Type Aftershock Sequence (ETAS) model by incorporating a renewal main-shock arrival process, which we term the renewal ETAS (RETAS) model. This modification is similar in spirit to the renewal Hawkes (RHawkes) process but the conditional intensity process supports a spatial domain. It empowers the main-shock intensity with the capability to reset upon the arrival of main-shocks and therefore allows for heavier clustering of earthquakes than the spatiotemporal ETAS model introduced by \cite{Ogata1998}. We introduce a likelihood evaluation algorithm for parameter estimation and provide a novel procedure to evaluate the fitted model's goodness-of-fit based on a sequential application of the Rosenblatt transformation. A simulation algorithm for the RETAS model is developed and applied to validate the numerical performance of the likelihood evaluation algorithm and goodness-of-fit test procedure. We illustrate the proposed model and procedures 
on various earthquake catalogs around the world each with distinctly different seismic activity. These catalogs will demonstrate that the RETAS model affords additional flexibility in comparison to the classical spatiotemporal ETAS model and has the potential for superior modeling and forecasting of seismicity.

\end{abstract}
\section{Introduction}
\label{sec:Introduction}
Point processes are a popular framework 
to describe patterns
of points that cluster in time, space, or a combination
thereof. Spatiotemporal point processes are frequently used to model
the clustering phenomena in seismic activity. 
The Epidemic-Type Aftershock Sequence (\etas) model has received
considerable attention for analyzing and forecasting the arrival times
and epicenters of earthquakes. \cite{Ogata1988} developed the first 
\etas\ model and then generalized it
to a spatiotemporal model \citep{Ogata1998}. The main-shock arrival process
for the \etas\ model is a spatiotemporal Poisson process
that is temporally homogeneous but spatially inhomogeneous, and the
aftershocks are triggered according to a parametric kernel function
that depends on the time, location and magnitude of the triggering
earthquake. The ETAS model has been the cornerstone for several related
works in 
seismology such as \cite{Console2001, Console2003,Zhuang2002,
	Zhuang2004, Zhuang2008, Ogata2004, Zhuang2005, Console2006,
	Marzocchi2009, Helmstetter2006, Werner2011}.

The 
success of self-exciting point processes in the
seismological literature is due primarily to the convenient branching
process interpretation introduced by \cite{Hawkes1974}, which
interprets the points of the process as either \textit{immigrants} or
\textit{offspring} events. 
In the \etas\ model, the immigrant events are interpreted as
main-shocks, which can
induce subsequent offspring events or aftershocks, which may
themselves induce additional aftershocks of their own
. 
The originating main-shock and the sequence 
of aftershocks induced by it form a cluster. The stability conditions in 
the model formulation of the \etas\ model guarantee that the cluster 
eventually ceases.  

Generally, most of the modeling flexibility in the \etas\ model is derived 
from the self-excited 
part
of the model. It is from the self-exciting 
part of the \etas\ model that credible forecasts of aftershocks can be
obtained. However, an inadequate description of the background
component can lead to incorrect classification of main-shocks and
aftershocks 
since the 
classification significantly depends on the specified background rate
(both temporally and spatially). 
This misspecification can also cause the Omori power-law
decay parameter to be either over or under estimated, and therefore the
aftershock sequence would appear shorter or more prolonged than it 
should be. 

Since the main-shock arrival process significantly
influences the clustering behaviour of an earthquake catalog, we
suggest that a more flexible choice for the background rate will
lead to 
a more appropriate distribution 
(and identification)
for main-shocks and their associated aftershock sequences. For
instance, in the New Zealand earthquake catalog to be analyzed
in Section~\ref{sec:fits-NZ}, the classical spatiotemporal \etas\ model
appears to be inadequate, and a modification to the model is required 
to fit the catalog. 
One approach to modifying the spatiotemporal \etas\ model of 
\cite{Ogata1998} is to allow the background seismicity to vary both 
temporally and spatially. 
Allowing the background rate to vary temporally has been discussed
previously in the works of \cite{Chen2013} and \cite{Godoy2016}.
Specifically, in the context of seismicity modelling, the nonstationary
\etas\ model introduced by \cite{Kumazawa2014} accounts for temporal 
variation in the main-shock arrival rate. In
their work, a time-dependent factor was introduced to the background
rate in the form of a first-order spline function to account for
systematic changes in sesmicity at different times. 
However, this
modification departs from the class of stationary point processes
and therefore may not be suitable for modelling long term
seismicity. As such, our proposed extension to the \etas\ model 
can model the short-term transient variations in seismicity while
maintaining long-term stationarity. 

Although nearly all stationary models assume a constant background
event rate in time, generally main-shocks do not occur uniformly in
time
. For instance, the accumulation and release of strain and stress may affect
the short-term main-shock rate. 
Therefore, it is highly plausible that main-shocks do not merely
reflect a spatiotemporal Poisson process that is homogenous in time
and spatially inhomogeneous.  Alternatively, the main-shock arrival
times might be better described by a general renewal process, which
allows the short term main-shock rate to vary according to the lapsed
time since the last main-shock, while still maintaining a constant
main-shock rate in the long term. 

Recently, in the work of \cite{Wheatley2016}, the self-exciting Hawkes
process was generalized to 
model 
the immigrants or background events 
using
a general renewal process but 
focused purely on the
temporal characteristics of the data.  \cite{Chen2018} then applied
the renewal Hawkes (RHawkes) process to an earthquake catalog from 
the Pacific Ring of Fire near the east coast of Japan (originally studied in
\cite{Ogata1988}). They reported that the improvement in fit based on
the renewal main-shock arrival process was adequate to describe the
temporal characteristics of the catalog. \cite{Kolev2018} applied an
analogous modification to the temporal \etas\ model by assuming the
waiting time 
distribution 
between main-shocks follows a gamma or
Brownian passage time (BPT) distribution. Their model was illustrated
on two earthquake catalogs from New Madrid and Northern
California. The non-Poissonian main-shock arrival models
studied thus far have focused mainly on the temporal aspects of the
data and they do not account for the spatial dependencies
that exist among earthquakes.

In this work we extend the classical spatiotemporal \etas\ model of
\cite{Ogata1998} by modelng the main-shoock arrival process as a renewal 
process which accommodates short term deviations of the
background seismicity from the long term mean level, while maintaining
the overall stationarity of the model. 
The process, termed the renewal \etas\ model ({\retas}), 
nests the \etas\ model and therefore allows for a determination of
whether the constant (temporal) background seismicity model 
is appropriate 
for particular seismic sequences. 

With the introduction of the renewal process for the background
seismicity, the main-shock rate depends on the lapsed time since the
previous main-shock which is not observable. This makes likelihood
evaluation challenging. However, using a recursive algorithm motivated
by that of \cite{Chen2018} for the RHawkes process likelihood
evaluation, we can directly evaluate the likelihood of the \retas\
model. 
We can then fit the \retas\ model by minimizing the 
negative log-likelihood and obtain the variance estimate for the
maximum likelihood estimator (MLE) by inverting the 
Hessian matrix. For goodness-of-fit (GOF) assessment of the \retas\ model, 
we propose a novel approach based on the Rosenblatt residuals 
\citep{Rosenblatt1952}, which can also be applied to the classical 
spatiotemporal \etas\ model. The approach avoids the simulations 
required by the thinning spatial residuals based approach of 
\cite{Schoenberg2003} and can assess the GOF of the temporal and 
spatial components of the model either simultaneously or separately. 

The remainder of this article is structured as follows. The general
form of the \retas\ model will be described in
Section~\ref{sec:Model}. Following this, the likelihood of the model
is discussed in detail in Section~\ref{sec:mle}. The method to assess
the GOF of the temporal and spatial aspects of the fitted model are
presented in
Section~\ref{sec:Model-Evaluation}. Section~\ref{sec:Simulations}
discusses the simulation of the {\retas} model and reports the results
of a simulation study to investigate the performance of the MLE and GOF test procedure. In
Section~\ref{sec:applications}, we illustrate the \retas\ model by
fitting various forms of the model to different earthquake catalogs 
around the world.

\section{General form of {\retas} model}
\label{sec:Model}
For an earthquake
catalog, we let $\set{(\tau_i,x_i,y_i,m_i)}_{i\ge1}$ denote the
occurrence time $\tau_i$, location $(x_i,y_i)$ and magnitude $m_i$
of each earthquake. Let $N$ be the point process for the earthquakes
with $N(A)$ counting the number of earthquakes in the set
$A \subset [0,T] \times \mathcal{S}\times \mathcal{M}$ where
$T \in \mathbb{R}$, $\mathcal{S} \subset \mathbb{R}^2$, and
$\mathcal{M}\subset \mathbb{R}$. Denote the internal history of $N$ by
$\mathcal{H}=\set{\mathcal{H}_t}_{t\in[0,T]}$, where
$\mathcal{H}_t = \sigma\set{(\tau_i,x_i,y_i,m_i) \, ; \, \tau_i \le
	t}$ is the sigma-field representing the knowledge of all the times,
locations and magnitudes of earthquakes up to and including time $t$.
The intensity process of $N$ relative to $\mathcal{H}$, or the
conditional event rate at time $t \in [0,T]$, location
$(x,y) \in \mathcal{S}$, and magnitude $m\in\mathcal{M}$ given the
internal history of the process prior to time $t$, $\mathcal{H}_{t-}$, is
defined as
\begin{multline}
	\label{eq:intensity-def}
	\lambda(t,x,y,m| \mathcal{H}_{t-}):= \\
	\lim\limits_{\Delta t, \Delta x,
		\Delta y, \Delta m \rightarrow 0}
	\frac{\E{N \brackets{[t,t+\Delta t) \times [x,x+\Delta x) \times
				[y,y+\Delta y) \times[m,m+\Delta m)}  | 
			\mathcal{H}_{t-}}}{\Delta t \Delta x \Delta y \Delta m}.
\end{multline} 

It is often assumed that 
the magnitudes of the earthquakes do not depend on their occurrence
times and locations, or the times, locations, and magnitudes of
previous earthquakes \citep{Zhuang2002}, and are independent and
identically distributed (\iid) with common density function $J$.
Therefore, if we let $\lambda_g$ denote the ground intensity process
for the times and locations of the earthquakes, that is
\begin{align*}
	\lambda_g(t,x,y)=  \lim\limits_{\Delta t, \Delta x,\Delta y \rightarrow 0}
	\frac{\E{N \brackets{[t,t+\Delta t) \times [x,x+\Delta x) \times
				[y,y+\Delta y) \times \mathcal{M}}  | 
			\mathcal{H}_{t-}}}{\Delta t \Delta x \Delta y },
\end{align*}
then the conditional intensity process of the marked
spatiotemporal point process 
takes the form
\begin{equation}
	\label{eq:intensity-txym}
	\lambda(t,x,y,m | \mathcal{H}_{t-}) = \
	\lambda_g(t,x,y | \mathcal{H}_{t-}) J(m).
\end{equation}

The \retas\ process intensity can be specified by introducing the
(unobservable) main-shock indicator $B_i$ such that $B_i = 0$ if the
$i$-th earthquake is a main-shock and $B_i = 1$ if it is an
aftershock. Furthermore, the function
$I(t) := \max \set{i \mid \tau_i < t, B_i = 0}$ indicates the index of
the most recent main-shock prior to time $t$ with the convention that
$I(t)=0$ for $t\leq \tau_1$ and $\tau_0=0$. For simplicity, 
we assume that there are
no events before time zero and the first event is a main-shock.
We introduce the extended history
$\tilde{\mathcal{H}_t} = \sigma(\mathcal{H}_t \cup \set{I(s); s \leq t})$
to include the index of the most recent main-shock (which is
unobservable from the earthquake catalog). The ground intensity of
the {\retas} model with respect to the extended history
$\tilde{\mathcal{H}}$ is then specified as
\begin{align}
	\label{eq:intensity-{\retas}}
	\tilde{\lambda}_g(t,x,y | \tilde{\mathcal{H}}_{t-})
	&
	= \mu(t-
	\tau_{I(t)})
	\nu(x,y) +
	\sum_{i:\tau_i <
		t}     
	\kappa(m_i) g(t - \tau_i) f(x - x_i, y - y_i)  \notag \\
	& = \mu(t- \tau_{I(t)}) \nu(x,y) + \phi(t,x,y).
\end{align}
The function $\mu: \mathbb{R}_{\geq 0} \rightarrow \mathbb{R}_{\geq 0}$ 
describes
the temporal variation of the main-shock arrival rate that renews 
upon the arrival of a main-shock. The spatial
variation of main-shocks over the region $\mathcal{S}$ is described by
the function $\nu: \mathcal{S} \rightarrow \mathbb{R}_{\geq 0}$. For
stability of the main-shock arrival process, it is assumed that
$
\int_{0}^{\infty} \exp (- \int_{0}^{t} \mu(s) \rmd s)\rmd t
< \infty,
$
which implies that the expected waiting time between consecutive
main-shocks is finite. For identifiability, it is also assumed that
$\iint_{\mathcal{S}} \nu(x,y) \rmd y\rmd x = 1$. The temporal response
function $g: \mathbb{R}_{\geq 0} \rightarrow \mathbb{R}_{\geq 0}$ 
represents the
temporal distribution of an aftershock relative to its triggering
earthquake, and the spatial response function
$f:\mathcal{S} \rightarrow \mathbb{R}_{\geq 0}$ describes the spatial
distribution of an aftershock relative to its triggering
earthquake. The functions $g$ and $f$ are both assumed to integrate to
unity. The boost function
$\kappa: \mathcal{M} \rightarrow \mathbb{R}_{\geq 0}$ represents the
influence of an earthquake on the instantaneous shock rate with
$\kappa(m)$ denoting the mean number of triggered aftershocks by an
earthquake with magnitude of size $m$.

The waiting times between main-shocks in the \retas\ model are
\iid\ with a common hazard function $\mu$, and therefore the main-shock
arrivals form a renewal process. When the function $\mu$ is a
constant, the \retas\ model~\eqref{eq:intensity-{\retas}} reduces to
the \etas\ model of \cite{Ogata1998}. The ground intensity of the
\retas\ process relative to the internal history $\mathcal{H}$ can be
expressed in the form
\begin{align}
	\label{eq:int-natural}
	{\lambda}_g(t,x,y | \mathcal{H}_{t-})
	&= \E{\mu(t- \tau_{I(t)})|\mathcal{H}_{t-}} \nu(x,y) + \sum_{i:\tau_i < t}    
	\kappa(m_i) g(t - \tau_i) f(x - x_i, y - y_i)  \notag \\
	& = \set{\sum_{j=1}^{N(t-)}\mu(t- \tau_{j})\P(I(t)=j|\mathcal H_{t-})} \nu(x,y) + \phi(t,x,y).
\end{align}
The ground intensity relative to the internal history is not 
readily available in general, unless $\mu$ is a constant,
and this has implications for likelihood evaluation.

For the remainder of this article, the following specific form
for the {\retas} model will be used.  The main-shock arrival process
is either Weibull with hazard function
$\mu_W(t) = \frac\alpha\beta (\frac t \beta) ^{\alpha - 1}, \ t\geq 0$,
or gamma with hazard function $\mu_G(t) =
\frac1{\beta\Gamma(t/\beta,\alpha)}(\frac t \beta)^{\alpha-1} 
e^{-t/\beta}$, $t \geq 0, $
where $\alpha>0$ and $\beta>0$ are the shape and scale parameters
respectively, and $\Gamma(x,k)=\int_x^\infty s^{k-1} e^{-s}\rmd s$
is the upper incomplete gamma function. The bivariate density
function $\nu(x,y)$ is learned from historical data
non-parametrically using a kernel density estimator
. The temporal response function $g$ takes the form of the modified
Omori's law \citep{Omori1894,Utsu1961},
\begin{math}
	g(t) = \frac{p-1}c(1+\frac t c)^{-p}, \ t \ge  0,
\end{math}
where $p > 1$ is a shape parameter indicating the decay rate of
aftershocks and $c>0$ is a scale parameter. The spatial response
function $f$ takes the form of a bivariate normal
density, 
\begin{math}
	f(x,y) = \frac{1}{2 \pi \sigma_1\sigma_2}\exp (-\frac{x^2}{2\sigma_1^2} - \frac{y^2}{2\sigma_2^2}),
\end{math}
where $\sigma_1 >0$ and $\sigma_2>0$ control the dispersion of the aftershock 
distribution in the $x$ and $y$ directions respectively. 
The boost
function 
takes the form,
\begin{math}
	\kappa(m) = Ae^{\delta(m - m_0)},
\end{math}
where $A > 0$ and $\delta>0$ control the expected number of induced aftershocks, and $m_0$ is the minimum or threshold magnitude of earthquakes in
the catalog. The 
distribution of the magnitudes are given by the 
the Gutenberg-Richer law \citep{Gutenberg1944}, which is a shifted
exponential distribution with density, 
\begin{math}
	J(m) = \gamma e^{-\gamma(m-m_0)}, \, m\geq m_0.
\end{math}


\section{Maximum Likelihood Estimation}
\label{sec:mle}
The difficulties that emerge from the complex expression for the
intensity process with respect to the internal history implies that
estimation of model parameters for the {\retas} model cannot be
readily performed using the classical likelihood formula for point
process models directly. However, in this section, we develop an
algorithm that computes the likelihood recursively,
therefore enabling likelihood-based inferences such as maximum
likelihood estimation
(MLE). 
For marked spatiotemporal point processes with conditional intensity
process defined as in~\eqref{eq:intensity-txym}, that is,
with respect to the internal history, the log-likelihood
can be represented in terms of the ground intensity process and an
additional term relating to the marks
\begin{equation}
	\label{eq:pp-loglik} 
	\ell(\theta) = \sum_{i=1}^{n} \log \brackets{\lambda_g(\tau_i,x_i,y_i)} -
	\int_{0}^{T} \iint_{\mathcal{S}} \lambda_g(t,x,y) \rmd y \rmd x \rmd t  
	+ \sum_{i=1}^n \log \brackets{J(m_i)} ,
\end{equation}
where $\theta = (\alpha, \beta, p, c, d, A, \delta, \gamma)^\top$
denotes the vector of parameters in the {\retas} model,
and $n=N([0,T]\times \mathcal{S}\times \mathcal{M})$
denotes the number of earthquakes in the study region in the interval $[0,T]$.

Since the distribution of the most recent main-shock index $I(t)$
given the prior-$t$ history $\mathcal{H}_{t-}$ is not readily
available, the loglikelihood expression in~\eqref{eq:pp-loglik} is not
amenable for direct evaluation. Therefore, we express the
log-likelihood in an alternative form and propose an algorithm for its
evaluation given in the subsequent theorem.
For convenience, we use the following notations; $s_i = (x_i,y_i)$, 
$p_{ij}=\P(I(\tau_i)  = j|\mathcal{H}_{\tau_i -})$, 
$d_1  = \mu(\tau_1) \nu(s_1) e^{{-\int_{0}^{\tau_1}\mu(t) \rmd t }}$, $d_{ij}  = \brackets{\mu(\tau_i - \tau_j)\nu(s_i) + \phi(\tau_i,s_i) } S_{ij}$, 
with $S_{ij}=e^{-\int_{\tau_{i-1}}^{\tau_i} \mu(t-\tau_j)  \rmd t},$ for $j=1,\dotsc,i-1$, $i=2,\dotsc,n$,  $S_{n+1,j}  = e^{-\int_{\tau_n}^{T}     \mu(t - \tau_j) \rmd t}$, for $j=1,\dotsc,n$, 
and, with $\phi(t,x,y)$ as in~\eqref{eq:intensity-{\retas}},
\begin{align*}
	\Phi(T) & =
	\int_0^T \iint_\mathcal{S} \phi(t,x,y) \rmd x\rmd y \rmd t \\
	& =
	\sum_{j=1}^n    \kappa(m_j) \int_{\tau_j}^T g(t - \tau_j)\,\rmd t\iint_{\mathcal{S}}f(x - x_j,y-y_j)\,\rmd x\rmd y.
\end{align*}

\begin{theorem}\label{th:loglik}
	The log-likelihood of the \retas\ model with 
	conditional intensity with respect to the extended 
	history $\tilde{\mathcal{H}}$ in~\eqref{eq:intensity-{\retas}} 
	is given by
	\begin{equation}
		\label{eq:log-lik}
		\ell(\theta ) = \log d_1 + \sum_{i=2}^{n} 
		\log \brackets{ \sum_{j=1}^{i-1}p_{ij}\,d_{ij} } + 
		\log \brackets{\sum_{j=1}^{n} p_{n+1,j} S_{n+1,j}  } + \Phi(T) + 
		\sum_{j=1}^{n} \log J(m_j),
	\end{equation}
	and the $p_{ij}$'s are recursively calculated using
	\begin{equation}
		\label{eq:pij-rec}
		p_{ij}  = \begin{cases}
			\frac{
				p_{i-1,j} \phi(\tau_i,s_i) S_{i - 1,j}}{
				\sum_{j=1}^{i-2} p_{i-1,j}d_{i-1,j}}, & j = 1,\dotsc,i - 2 \\
			1-\sum_{k = 1}^{i - 2}p_{ik}, & j = i - 1
		\end{cases},\quad i\geq 3,
	\end{equation}
	from the initial value $p_{21}=1$. 
\end{theorem}

\begin{proof}
	For notational convenience, for a sequence $z_1,z_2,\dotsc$, we use
	$z_{i:j}$ to denote $z_i,z_{i+1},\dotsc,z_j$, $i<j$. The likelihood
	function for the {\retas} model is calculated by computing the
	sequential contribution to the likelihood between successive
	earthquakes as follows
	\begin{equation}
		L(\tau_{1:n}, s_{1:n}, m_{1:n}) = 
		p(\tau_1,s_1,m_1) 
		\braces{ \prod_{i=2}^{n} p(\tau_i,s_i,m_i | \mathcal{H}_{\tau_i-}) } 
		\P \brackets{\tau_{n+1} > T | \mathcal{H}_T}.
	\end{equation}
	Then by invoking the independent magnitude assumption, the
	likelihood can be separated into a spatiotemporal component
	for the ground intensity process and the marks as follows
	\begin{multline}
		L(\tau_{1:n}, s_{1:n}, m_{1:n}) = 
		p(\tau_1,s_1) p(m_1 | \tau_1,s_1)\times {}  \\
		\braces{ \prod_{i=2}^{n} p(\tau_i,s_i | \mathcal{H}_{\tau_i-}) 
			p(m_i | \tau_i,s_i, \mathcal{H}_{\tau_i-})} 
		\P \brackets{\tau_{n+1} > T | \mathcal{H}_T}.
	\end{multline}  
	Therefore, the log-likelihood $\ell(\theta)$ is
	separable into a spatiotemporal component $\ell_s(\theta)$ and a
	likelihood contribution for the magnitudes $\ell_m(\theta)$, in
	which
	\begin{equation}
		\label{eq:ls}
		\ell_s(\theta) = \log p(\tau_1, s_1) + \sum_{i=2}^{n} 
		\log  p(\tau_i,s_i | \mathcal{H}_{\tau_i-}) \\ + 
		\log    \P \brackets{\tau_{n+1} > T | \mathcal{H}_{T}}
	\end{equation}
	\normalsize
	and
	\begin{equation}
		\ell_m(\theta) = \log p(m_1 | \tau_1, s_1) + \sum_{i=2}^{n} \log p(m_i
		| \tau_i,s_i,\mathcal{H}_{\tau_i-})=\sum_{i=1}^n \log J(m_i). 
	\end{equation}
	The form of $\ell_m$ is complete and the remainder of this proof deals with the evaluation of $\ell_s$. 
	By conditioning on the index of the most recent main-shock,
	the 
	expression in~\eqref{eq:ls} has the equivalent form
	\begin{multline*}
		\ell_s(\theta) = \log p(\tau_1, s_1)  
		+ \sum_{i=2}^{n} 
		\log \brackets{ \sum_{j=1}^{i-1} p(\tau_i,s_i | I(\tau_i) = j, \mathcal{H}_{\tau_i-}) 
			p_{i,j} }\\
		+ 
		\log \sum_{j=1}^{n}      \P \brackets{I(\tau_{n+1}) = j | \mathcal{H}_{T}}
		p_{n+1,j} ,
	\end{multline*}
	where
	\begin{align*}
		& p(\tau_i,s_i | I(\tau_i) = j, \mathcal{H}_{\tau_i-}) \\
		& \qquad = (\mu(\tau_i - \tau_j) \nu(s_i) + \phi(\tau_i,s_i))
		\exp \brackets{-\int_{\tau_{i-1}}^{\tau_i} \iint_{\mathcal{S}} \mu(t - \tau_j) \nu(s)
			+ \phi(t,s) \, \rmd s \rmd t   }\\
		&\qquad = d_{ij}\exp\brackets{-\int_{\tau_{i-1}}^{\tau_i} \iint_{\mathcal{S}} \phi(t,s) \, \rmd s \rmd t},
	\end{align*}
	and
	\begin{align*}
		\P \brackets{I(\tau_{n+1}) = j | \mathcal{H}_{T}}  &=
		\exp \brackets{-\int_{\tau_n}^T \iint_{\mathcal{S}} \mu(t -
			\tau_j)\nu(s) + \phi(t,s) \, \rmd s \rmd t   }\\ 
		&
		= S_{n+1,j}\exp\brackets{-\int_{\tau_{n}}^{T} \iint_{\mathcal{S}} \phi(t,s) \, \rmd s \rmd t}.
	\end{align*}
	Therefore, 
	it follows that
	\begin{equation}
		\ell_s(\theta ) = \log d_1 + \sum_{i=2}^{n} 
		\log \brackets{ \sum_{j=1}^{i-1}p_{ij}\,d_{ij} } + 
		\log \brackets{\sum_{j=1}^{n} p_{n+1,j} S_{n+1,j}  } + \Phi(T).
	\end{equation}
	To complete the proof, it remains to show the
	recursion~\eqref{eq:pij-rec} for the distribution of the most recent
	main-shock index. The proof is analogous to Eq.~(S.6)-(S.9)
	in the file proofs.pdf contained in the online supplementary
	material accompanying \cite{Chen2018}
	. \qed
\end{proof}

The MLE of the model parameter is obtained by directly maximizing
the log-likelihood function in~\eqref{eq:log-lik} using general purpose
optimization routines, and inverting the Hessian matrix to estimate the
variance of the MLE. For this purpose, fast and accurate
computation of the double integral
$\iint_\mathcal{S}f(x-x_j,y-y_j)\rmd x\rmd y$ required
in~\eqref{eq:log-lik} is important. When the spatial domain
$\mathcal{S}$ is large relative to the variance of the spatial response
function $f$, the integral can be approximated by 1. 
However, in general this approximation can be too crude and lead to
substantial bias in the estimators. Numerical quadrature can be used
to approximate the integral accurately, but 2-dimensional (2-d)
quadrature can be intolerably slow. For some choices of the function $f(x,y)$,
such as those that depend on the input $(x,y)$ only through its norm
$\sqrt{x^2+y^2}$, the 2-d integral can be reduced to 1-d integrals
through a change of variables using polar coordinates and then
evaluated using 1-d quadrature.

\section{Goodness-of-fit Assessment}
\label{sec:Model-Evaluation}
An assessment of the adequacy of the fit of the \retas\ model to an
earthquake catalog should be conducted once the fitted model has been 
obtained. 
There are
two obvious aspects of the fitted model that should be assessed, that
is, the temporal and the spatial aspects. These two aspects are
typically assessed separately using different approaches, with the
assessment of the temporal aspect based on the Papangelou time-change
theorem \citep[Theorem~7.4.I][]{Daley2003}, and the assessment of the
spatial aspect based on a thinning technique
\citep{Schoenberg2003}. However, the intensity process of the \retas\
model relative to the internal history is cumbersome to work with,
which makes it inconvenient to use these approaches. Instead, we propose a
unified approach to evaluate the goodness-of-fit (GOF) of both aspects
of the \retas\ model based on the \cite{Rosenblatt1952}
transformation.

For a $n$-dimensional random vector $Z = (Z_1,\dotsc,Z_n)$, the
Rosenblatt transformation $T(Z)=(T_1(Z),\dotsc,T_n(Z))$ is defined by
$T_1(z)=\P(Z_1\leq z_1)=F_1(z_1)$, and
$T_i(z)=\P(Z_i\leq
z_i|Z_1=z_1,\dotsc,Z_{i-1}=z_{i-1})=F_i(z_i|z_1,\dotsc,z_{i-1})$ for
$i=2,\dotsc,n$, where the $F_i$'s are the conditional distribution
functions. It is known that the distribution of $T(Z)$ is uniform on
the hypercube $[0,1]^n$. When the conditional distribution functions
are replaced with their estimators $\hat F_i$'s, then the Rosenblatt
residuals $\hat F_1(Z_1),\dotsc,\hat F_n(Z_n|Z_1,\dotsc,Z_{n-1})$
are approximately \iid\ uniform on $[0,1]$ when the model is
correctly specified. Uniformity can be checked using formal
statistical tests such as the Kolomogorov-Smirnoff (K-S) test and
independence by the Ljung-Box (L-B) test or visually assessed using
graphical tools such as the uniform Q-Q plot and the ACF plot.

\subsection{Temporal model assessment}

Specific to the \retas\ model, the Rosenblatt residuals to assess
the temporal component of the model are given by
$U_i = \hat F_i(\tau_i| \mathcal H_{\tau_{i-1}})$ where $\hat F_i$ is
the fitted conditional distribution function of $\tau_i$ given
$\mathcal{H}_{\tau_{i-1}}=\sigma\braces{\tau_{1:i-1},
	x_{1:i-1},y_{1:i-1},m_{1:i-1}}$. Let $\hat\mu(\cdot)$ be the plugin
estimate of $\mu(\cdot)$ with the unknown parameters replaced with
their estimted values, and
$\hat \kappa(\cdot), \hat g(\cdot), \hat f(\cdot,\cdot)$, and
$\hat p_{ij}$ be analogously defined. Then $U_1=\hat F_1(\tau_1)=1- \exp\brackets{-\int_{0}^{\tau_1} \hat\mu(s) \rmd
	s}$
and for $i=2,\dotsc,n$,
\begin{align*}
	U_i&=\hat F_i(\tau_i|\mathcal H_{\tau_{i-1}})= 1 - \sum_{j=1}^{i-1}
	\hat p_{ij} \hat{S}_{ij},
\end{align*}
where
\begin{displaymath}
	\hat{S}_{ij} =
	\exp \bigg(-\int_{\tau_{i-1}}^{\tau_i} \bigg[\hat \mu(s - \tau_j) 
	+ \sum_{k=1}^{N(t-)} \hat \kappa(m_i)\hat g(s-\tau_k)\iint_{\mathcal{S}} \hat f(x-x_k,y-y_k) \rmd y \rmd x \bigg] \rmd s\bigg).
\end{displaymath}


\subsection{Spatial model assessment}
For spatial model assessment, two residuals $V_i$
and $W_i$ are calculated for the longitude $x_i$ and latitude $y_i$ of the
earthquakes respectively. First, we calculate the longitudinal residual
$V_i$ conditional on the time $\tau_i$ of the quake and the history
$\mathcal H_{\tau_i-}$ of the process prior to $\tau_i$, and then
we calculate the latitudinal residual $W_i$ with the longitude $x_i$ of
the quake also included in the conditioning
information. 
That is, we define the longitudinal
residual 
$V_i = \hat G_i(x_i| \tau_i,\mathcal H_{\tau_i-})$,
where $\hat G_i$ denotes the fitted conditional distribution of $x_i$
given $\tau_i$ and $\mathcal{H}_{\tau_i-}$, and the latitudinal
residual 
$W_i = \hat H_i(x_i| \tau_i,x_i,\mathcal{H}_{\tau_i-})$, where
$\hat H_i$ denotes the fitted conditional distribution of $y_i$ given
$\tau_i, x_i,$ and $\mathcal{H}_{\tau_i-}$. 
In the following, when presenting the specific expressions of the residuals,
we suppress the hat notation ($\hat{\phantom{\nu}}$) from various estimated
parameters for convenience.

The longitudinal residual is given, for $i=1$, by
$
V_1 = \iint_{ \mathcal{S} \cap \set{(x,y); \, x \le x_i}} \nu(x,y)\rmd
y \rmd x,\label{eq:V1}
$
and for $i =2, \dotsc, n$, by
\begin{equation}\label{eq:Vi}
	V_i = \sum_{j=1}^{i-1} \biggl\{ \frac{\mu(\tau_i - \tau_j)  \iint_{\mathcal{S} \cap \set{(x,y); \, x \le x_i}} \nu(x,y) \rmd y \rmd x
		+  \iint_{\mathcal{S} \cap \set{(x,y); \, x \le x_i}} \phi(\tau_i,x,y) \rmd y \rmd x  }
	{ \mu(\tau_i - \tau_j) + \iint_{\mathcal{S}} \phi(\tau_i,x,y) \rmd y\rmd x} p^\tau_{ij}
	\biggr\},
\end{equation}
where
\begin{equation*}
	\iint_{\mathcal{S}} \phi(\tau_i,x,y) \rmd y \rmd x = 
	\sum_{k=1}^{i-1} \kappa(m_k) g(t-\tau_k) \iint_{\mathcal{S}} f(x - x_k,y - y_k) \rmd y \rmd x,
\end{equation*}
and
$p_{ij}^{\tau} = \P\brackets{I(\tau_i) = j | \tau_i,\mathcal{H}_{\tau_i-}}$ 
are the most recent main-shock probabilities that are updated by 
including $\tau_i$ in the condition,
which differ slightly from the $p_{ij}$ in~\eqref{eq:pij-rec}. 
By updating the most recent main-shock probabilities
$p_{ij}$ computed in the likelihood evaluation, we see that
\begin{equation}
	\label{eq:pij.tau-rec}
	p_{ij}^\tau = \frac{p(\tau_i |I(\tau_i) = j,\mathcal{H}_{\tau_i-}) }
	{p(\tau_i | \mathcal{H}_{\tau_i-})}\, p_{ij},
\end{equation}
where 
\begin{equation*}
	p(\tau_i | I(\tau_i) = j,\mathcal{H}_{\tau_i-}) = \lambda_\tau^j(\tau_i) \exp \brackets{-\int_{\tau_{i-1}}^{\tau_i} \lambda_\tau^j(t) \rmd t},
\end{equation*}
and
\begin{equation*}
	p(\tau_i |\mathcal{H}_{\tau_i-}) = \sum_{j=1}^{i-1} \lambda_\tau^j(\tau_i) \exp \brackets{-\int_{\tau_{i-1}}^{\tau_i} \lambda_\tau^j(t) \rmd t} \, p_{ij},
\end{equation*}
with
\begin{equation*}
	\lambda_\tau^j(t) 
	= \mu(t - \tau_j) + \sum_{k=1}^{N(t-)} \kappa(m_k)
	g(t-\tau_k)\iint_{\mathcal{S}} f(x - x_k,y - y_k) \rmd y \rmd x.
\end{equation*}

The latitudinal residual is defined similarly using the
conditional distribution of $y_i$ given $\mathcal{H}_{\tau_i-}$,
$\tau_i$ and $x_i$. For $i=1$, it is given by
$$
W_1 = \frac{\int_{\set{y | (x_1,y) \in \mathcal{S},\,  y \le y_i}}
	\nu(x_1,y) \rmd y}{ \int_{\set{y | (x_1,y) \in \mathcal{S}} }\nu(x_1,y)
	\rmd y},
$$
and for $i = 2,\dotsc,n$, by
\begin{equation}
	W_i = \sum_{j=1}^{i-1} \biggl\{ \frac{\mu(\tau_i - \tau_j)  \int_{\set{y | (x_i,y) \in \mathcal{S} ,\, y \le y_i}} \nu(x_i,y) \rmd y
		+  \int_{\set{y | (x_i,y) \in \mathcal{S} ,\, y \le y_i}} \phi(\tau_i,x_i,y) \rmd y}
	{  \mu(\tau_i - \tau_j)\int_{\set{y | (x_i,y) \in
				\mathcal{S} }} \nu(x_i,y) \rmd y +   
		\int_{\set{y | (x_i,y) \in \mathcal{S} }} \phi(\tau_i,x_i,y)\rmd y} \, p^x_{ij}
	\biggr\},
\end{equation}
where
$p_{ij}^{x} = \P\brackets{I(\tau_i) = j | \tau_{i},
	x_{i},\mathcal{H}_{\tau_i-}}$ are the most recent main-shock
probabilities, but now include both $\tau_i$ and
$x_i$ in the condition. The previously computed $p_{ij}^\tau$ in~\eqref{eq:pij.tau-rec}
are updated to compute $p_{ij}^x$ as follows
\begin{equation*}
	p_{ij}^x = \frac{p(x_i | \tau_{i}, I(\tau_i) = j,\mathcal{H}_{\tau_i-})}
	{p(x_i | \tau_{i}, \mathcal{H}_{\tau_i -})}\, p_{ij}^\tau,
\end{equation*}
where the densities are given by
\begin{equation*}
	p(x_i | \tau_{i},I(\tau_i) = j, \mathcal{H}_{\tau_i -}) \\
	= 
	\mu(\tau_i - \tau_j) \int_{\set{y | (x_i,y) \in \mathcal{S} }} \nu(x_i,y) \rmd y + 
	\int_{\set{y | (x_i,y) \in \mathcal{S} }} \phi(\tau_i,x_i,y) \rmd y,
\end{equation*}
and
\begin{equation*}
	p(x_i | \tau_{i}, \mathcal{H}_{\tau_i -}) \\
	=       \sum_{j=1}^{i-1} p(x_i | \tau_{i}, I(\tau_i) =
	j,\mathcal{H}_{\tau_i -})\, p_{ij}^\tau.
\end{equation*}

\section{Simulations}
\label{sec:Simulations}
This section reports the results of a simulation study for the
{\retas} model and confirms that the finite sample performances of the
MLE are as expected. 
We first outline a procedure to simulate the \retas\ model%
.

\subsection{Simulation algorithm for the \retas\ model}
\label{sec:sim-alg}
Simulation of the {\retas} model can be efficiently performed by
utilizing the branching structure of the process. The algorithm
operates as follows.
\begin{enumerate}
	\item Simulate the occurrence times of main-shocks up to the censoring
	time $T$, as the cumulative sum of \iid\ positive random variables
	with hazard rate function $\mu(\cdot)$, and denote these by
	${\tau_1^0,\dots,\tau_{n_0}^0}$.
	
	\item For each main-shock $i=1,\dots,n_0$ simulate the location
	according to the bivariate probability density function $\nu(\cdot)$
	and denote the locations as ${s_1^0,\dotsc,s_{n_0}^0}$.
	
	\item Next, simulate the main-shock magnitudes $m_i^0$ according to
	the density $J(\cdot)$ and store all the times, locations, and
	magnitudes of main-shocks as the generation 0 catalog
	$G^{(0)}=\set{(\tau^0_i,s^0_i,m^0_i),\, i=1,\dotsc,n_0}$. Now set
	$l = 0$.
	
	\item \label{item:iter}For each earthquake in the generation $l$ 
	catalog $G^{(l)}$, namely, $(\tau_k^{(l)},s_k^{(l)},m_k^{(l)})$ for $k=1,\dotsc, |G^{(l)}|$, 
	simulate 
	the potential number of aftershocks as a Poisson random variable
	with mean $\kappa(m_k^{(l)})$.  Then for each potential aftershock,
	simulate its occurrence time relative to $\tau_k^{(l)}$ according to the
	temporal response function $g(\cdot)$, and discard the aftershock if
	the simulated occurrence time is beyond $T$; if retained, simulate
	its location relative to $s_k^{(l)}$ according to the spatial response
	function $f(\cdot, \cdot)$, and discard the aftershock if the
	location is outside $\mathcal{S}$; and if retained, simulate its
	magnitude according to the density function $J(\cdot)$. Record the
	times, locations and magnitudes of all retained
	aftershocks 
	as the generation $l+1$ catalog $G^{(l + 1)}$.
	
	\item 
	If the catalog $G^{(l+1)}$ is non-empty, set $l \leftarrow l +1$
	and return to Step~\eqref{item:iter}; otherwise, collect the simulated
	earthquakes of all generations and return them as the overall
	catalog $G=\bigcup_l G^{(l)}$.
	
\end{enumerate}

\subsection{Simulation Results}
\label{sec:simulations}
In this section, numerical evidence confirms that the
simulation, 
likelihood evaluation, 
and 
goodness-of-fit test 
algorithms are performing correctly and establishes
that the MLE 
has satisfactory finite sample
performance. The {\retas} model examined in this simulation study is
the same as described in Section~\ref{sec:Model}, but here we
only investigate the Weibull hazard function $\mu_W(t)$. The
magnitudes $m$ are simulated using a shifted exponential distribution
with rate parameter $\gamma = 5$ and threshold magnitude $m_0 = 6$.

The simulations consists of $1000$ realizations of the \retas\ model up to the
censoring time $T=200$ over the whole plane, i.e.
$\mathcal{S} = \mathbb{R}^2$. The shape and scale parameters of the
Weibull hazard function are chosen to be $(\alpha,\beta)=(0.5,0.5)$ or
$(2,1)$ to generate both heavily (temporally) clustered main-shocks
and under-dispersed main-shocks. The scale parameter $\beta$ was
chosen to have a mean waiting time between consecutive
main-shocks close to unity. The main-shocks are spatially distributed
according to an independent-marginal bivariate normal distribution
with mean at the origin and standard deviations $0.25$ and $0.5$ in
the $x$ and $y$ directions respectively. The parameters in the
temporal response function $g$ are fixed
at 
$p=2$ and $c=0.01$. The spatial response function is 
a bivariate normal density function with independent marginals and
variances $\sigma_1 = 0.01$ and $\sigma_2 = 0.02$ in the $x$
and $y$ directions respectively. The parameters for the boost
function $\kappa$ are $A=0.5$ and $\delta=1$, which implies that an
earthquake 
will induce, on average,
$ \int_{m_0}^\infty Ae^{\delta(m-m_0)}\frac1\gamma
e^{-(m-m_0)/\gamma}\rmd m = 0.625$ aftershocks. With such choice of
the model parameters, the number of earthquakes in the first
simulation model ranges from 718 to 1503 with
mean 
$1063$, and the number of earthquakes in the second simulation model ranges from
1008 to 1405 with mean 1189. The more volatile number of
earthquakes in the first simulation model is to be expected, as the shape
parameter $\alpha=0.5$ for the Weibull distribution therein implies
more volatile waiting times between main-shock arrivals.

The results of the simulation study are reported in
Table~\ref{tab:simstudy}, which contains the true value of each
parameter (True), the mean of the $1000$ parameter estimates by
directly minimizing the negative log-likelihood function (Est), the
empirical standard error of the $1000$ parameter estimates (SE), the
estimated standard error by computing the mean of the standard errors
found by inverting the approximate Hessian matrix
$(\widehat{\text{SE}})$, and the empirical coverage probability (CP)
of the $95\%$ confidence interval obtained by assuming (asymptotic)
normality of the estimator.

\begin{table}
	\caption{\label{tab:simstudy} Estimation results for the simulated data using MLE for the \retas\ model over the region $\mathcal{S} = \mathbb{R}^2$ with Weibull main-shock renewal process, an aftershock density following Omori's law and bivariate normal aftershock spatial distribution.
	}	
	\centering
		\begin{tabular}[htb!]{l | llllllll}
			& $\alpha$ & $\beta$ & $p$ & $c$ & $\sigma_1$ & $\sigma_2$ &  $A$ & $\delta$ \\
			\hline \hline
			True & 0.5 & 0.5 & 2 & 0.01 & 0.01 & 0.02 & 0.5 & 1 \\
			Est &  0.5014 & 0.5197 & 2.0148 & 0.0104 & 0.0104 & 0.0218 & 0.5045 & 0.9756 \\
			SE & 0.0313 & 0.0899 & 0.1668 & 0.0026 & 0.0013 & 0.0042 & 0.0455 & 0.2737 \\ 
			$\widehat{\text{SE}}$ & 0.0307 & 0.0865 & 0.1783 & 0.0026 & 0.0011 & 0.0023 & 0.0474 & 0.2894 \\ 
			CP & 0.9539 & 0.9609 & 0.9529 & 0.9449 & 0.8898 & 0.8317 & 0.9539 & 0.9579 \\ 
			\hline \hline
			True & 2 & 1 & 2 & 0.01 & 0.01 & 0.02 & 0.5 & 1 \\
			Est &   1.9913 & 0.9986 & 2.0205 & 0.0104 & 0.0102 & 0.0210 & 0.4978 & 1.0134 \\ 
			SE & 0.1306 & 0.0408 & 0.1558 & 0.0023 & 0.0010 & 0.0029 & 0.0397 & 0.2438 \\ 
			$\widehat{\text{SE}}$ &  0.1218 & 0.0407 & 0.1554 & 0.0023 & 0.0009 & 0.0018 & 0.0415 & 0.2582 \\
			CP & 0.9280 & 0.9460 & 0.9460 & 0.9390 & 0.9180 & 0.8660 & 0.9530 & 0.9660
	\end{tabular}
\end{table}

The MLE estimator demonstrates unbiasedness as the
estimated parameters are, on average, close to their respective true
values. The empirical standard errors and the average of the standard
errors are very similar
. Furthermore, the 95\% coverage probabilities are all reasonably 
close to their
nominal level
. Therefore, we conclude that the MLE provides a satisfactory
finite sample performance, and the simulation algorithm presented in Section~\ref{sec:sim-alg} is simulating from the correct model 
specification. 

The appropriateness of the sequential goodness-of-fit test procedure 
introduced in Section~\ref{sec:Model-Evaluation} will now be assessed. 
To this end, the three residual series $\set{U}$, $\set{V}$, and
$\set{W}$ for each simulated catalog are computed under two
scenarios using the estimated parameters and the true parameters. 
For both situations, the residuals series are examined for
uniformity using the K-S test and independence using the L-B
test. The results are presented in Table~\ref{tab:gofstudy}, which
reports the percentage of $p$-values for the several tests that are
less than the nominal significance level as indicated on the left-hand
panel for both the 5\% and 1\% level, for the temporal,
longitudinal and latitudinal components.

\begin{table}
	\caption{\label{tab:gofstudy} Percentage of K-S and L-B tests which lead to rejection of the null hypothesis for the two simulation models based on 1000 simulated datasets by computing residuals that are evaluated at both the estimated parameter values of the true simulated parameter values.}
	\centering
		\begin{tabular}[htb!]{ll|ll|ll|ll}
			\hline
			& & \multicolumn{2}{c|}{U} &\multicolumn{2}{c|}{V} &\multicolumn{2}{c}{W} \\
			& & K-S & L-B & K-S & L-B & K-S & L-B \\ 
			\hline
			\multicolumn{8}{c}{Estimated parameter values} \\
			\hline
			\multirow{2}{*}{5\%} & Model 1 & 0.00\% & 6.01\% & 4.21\% & 6.51\% & 4.61\% & 5.91\% \\ 
			& Model 2 & 0.00\% & 6.40\% & 4.20\% & 8.00\% & 6.40\% & 5.70\% \\ 
			\hline
			\multirow{2}{*}{1\%} & Model 1 & 0.00\% & 2.10\% & 0.80\% & 2.20\% & 0.60\% & 1.20\% \\
			& Model 2 & 0.00\% & 1.70\% & 0.70\% & 2.20\% & 1.60\% & 1.40\% \\ 
			\hline
			\multicolumn{8}{c}{True parameter values} \\
			\hline
			\multirow{2}{*}{5\%} & Model 1 & 4.91\% & 6.01\% & 4.21\% & 7.11\% & 4.61\% & 5.81\% \\ 
			& Model 2 & 3.40\% & 6.40\% & 4.80\% & 8.00\% & 5.20\% & 5.50\% \\
			\hline
			\multirow{2}{*}{1\%} & Model 1 & 1.10\% & 2.10\% & 0.90\% & 2.40\% & 0.30\% & 1.60\% \\ 
			& Model 2 & 0.80\% & 1.70\% & 0.80\% & 2.20\% & 1.10\% & 1.20\% \\ 
			\hline
	\end{tabular}
\end{table}

The percentage of residual series that reject the uniformity or
independence hypothesis is relatively close to the nominal levels 
indicating that the GOF test procedure is performing as
expected at the given significance level. 
For the temporal component, the residuals series computed using 
the estimated parameters all pass the uniformity test even at the 
5\% level, which is to be expected since using the fitted
parameters with a correctly specified model under the null
hypothesis generally leads to inflated $p$-values in GOF tests. 
However, when the true parameters are used in place of the estimated
parameters, the percentage of failed tests is much closer to the
nominal levels. In conclusion, the GOF tests based on these residuals
appear to work well as expected. 



\section{Data Analysis}
\label{sec:applications}
Many regions around the world have significantly distinct and unique
seismic activity. The versatility of the \retas\ model will be
illustrated on various earthquake catalogs around the world,
including New Zealand, Chile, Japan, and China. These catlogs will confirm
that by introducing the renewal main-shock arrival process to the 
classical \etas\ model improves the GOF. Since
these two models are nested, the AIC is computed to establish
that the \retas\ model is the superior model of choice for these particular
earthquake catalogs (when the self-exciting structure for aftershocks 
and the spatial variation for main-shocks have the same specification). 

\subsection{New Zealand earthquake catalog}
\label{sec:fits-NZ}

In this subsection we investigate a large region that includes 
most of the seismically active regions of New Zealand. 
The earthquake catalog is complied from the GeoNet 
Quake Search Database. The region under investigation 
is defined by the coordinates $164^\circ-182^\circ$~N and 
$48^\circ - 35 ^\circ$~S which consists of the majority 
of New Zealand as seen in Figure~\ref{fig:nzqks}. This 
region has previously been studied in the work of 
\cite{Harte2013, Harte2014}. The catalog contains 463 
earthquakes during the period 1997-01-01 until 2015-06-30 
with threshold magnitude $m_0=5$. 
Figure~\ref{fig:nzqks} displays all earthquakes that have 
occurred since the 1800-01-01 until 2015-06-30, with the 
size of the circle indicating the relative size of the magnitude
of the earthquake. Using historical 
and in-sample data we estimate the main-shock (background) 
spatial density by applying a two-dimensional Gaussian kernel density 
smoother to all the earthquakes that occurred in the region.

\begin{figure}[htb!]
	\caption{\label{fig:nzqks} Earthquakes around New Zealand from 1800-01-01 until
		2015-06-30 with the size of the circle indicating the relative size
		of the magnitude of the earthquake.}
	\centering
	\includegraphics[width = 0.70\linewidth]{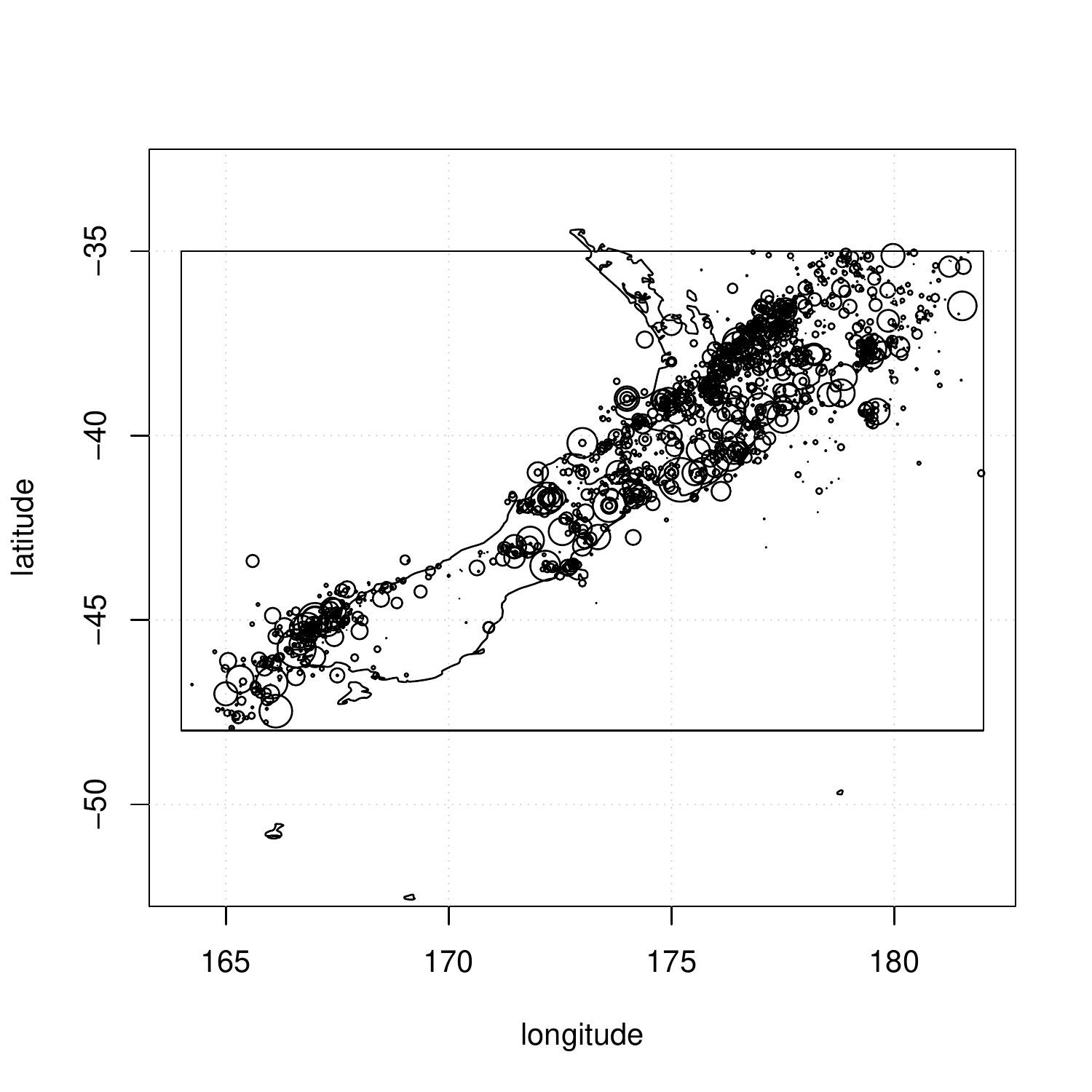}
\end{figure}

The renewal main-shock arrival process is fit with three different 
forms; exponential (classical ETAS), Weibull and gamma 
distributions. For each of the three main-shock renewal distributions, 
we report the estimation results in Table~\ref{tab:NZ-fits}, which 
contains the estimated model parameters, 
the standard errors 
and the AIC estimate for the fitted model. General purpose 
optimization routines are used to obtain estimates of 
model parameter with the initial parameters set using a fit 
from a nested model. More specifically, we fit 
sequentially the following models; Poisson process, temporal Hawkes process, 
\etas\ model, spatiotemporal ETAS model
and \retas\ model. \\


\begin{table}
	\caption{\label{tab:NZ-fits} Estimation results which includes parameter estimates and standard errors in parentheses for the New Zealand earthquake catalog with exponential (ETAS), Weibull and gamma RETAS models.}
	\centering
	\resizebox{\linewidth}{!}{
			\begin{tabular}{lccccccccc}
				\hline
				& $\hat\alpha$ & $\hat\beta$ & $\hat p$ & $\hat c$ & $\hat \sigma_1$ & $\hat \sigma_2$ & $\hat A$ & $\hat\delta$ & AIC  \\ 
				\hline \hline
				\multirow{2}{*}{\etas} & 1 &  25.16 & 1.095 & 0.0021 & 0.065 & 0.017 & 0.217 & 1.719 & \multirow{2}{*}{4599.21 } \\ 
				& (NA)& (1.66) & (0.025) & (0.00036) & (0.0088) & (0.0025) & (0.044) & (0.11) & \\ \hline
				\multirow{2}{*}{{\retas}$_{Wei}$} & 0.864 & 23.00 & 1.128 & 0.0026 & 0.063 & 0.017 & 0.179 & 1.708 & \multirow{2}{*}{4590.99} \\ 
				& (0.044) & (1.79) & (0.024) & (0.00064) & (0.0085) & (0.0024) & (0.027) & (0.11) & \\ \hline
				\multirow{2}{*}{{\retas}$_{Gam}$} & 0.752 & 33.425 & 1.092 & 0.0020 & 0.062 & 0.017 & 0.222 & 1.712 & \multirow{2}{*}{4588.95} \\ 
				& (0.060) & (3.68) & (0.024) & (0.00028) & (0.0084) & (0.0024) & (0.046) & (0.11) & \\ \hline
	\end{tabular}}
\end{table}

The flexibility provided by the \retas\ model is mainly derived from the parameters $\alpha$ and $\beta$ which describe the renewal main-shock arrival process and we examine these parameter first. 
The estimated shape parameters $\hat\alpha = 0.864$ and $\hat \alpha = 0.752$ for the Weibull and gamma \retas\ models are both significantly less than one with 95\% confidence intervals $(0.778, 0.950)$ and $(0.635,0.869)$ respectively. This implies that the main-shock arrival process depart from homoegnous temporal variation $(\alpha=1)$ and indicates that more substantial clustering is apparent than that of the classical \etas\ model. 
The mean waiting time between consecutive main-shocks for the three models are all very close with $25.11$, $24.78$ and $25.14$ days for the classical \etas, Weibull, and gamma {\retas} models respectively. However, the volatility of the main-shock waiting times are different. For instance, the standard deviation for these waiting times are 25.10 for the classical \etas\ model but are higher for the Weibull and gamma \retas\ model at 28.78 and 28.99, respectively. 
%

The self-excitation component of each of the three fitted models have
similar parameter estimates and infer comparable aftershock
properties. The estimated parameters $\hat p$ and $\hat c$ are 
comparable among the models. Take for example the gamma \retas\ model 
in which $\hat p =1.092$ and $\hat c = 0.0020$. This implies that the 
aftershock occurs with probability 24.59\% within the next hour, 
43.47\% within a day, 52.74\% within a week and 67.15\% within the 
next year after the triggering earthquake. The aftershocks generally 
occur shortly after the triggering earthquake but its effects can 
still last a long time thereafter. For instance, the median waiting 
time between an earthquake and its induced aftershock (if there is any) 
is 3 days and 19.13 hours for the gamma \retas\ model, and this is 
significantly shorter than the median waiting time between main-shocks 
which is around 10.25 days. The estimated parameters $\hat \sigma_1$ 
and $\hat \sigma_2$ are also very close for all three models with 
$\hat \sigma_1$ being more than three times larger than 
$\hat \sigma_2$. The estimated boost function parameters $\hat A$ 
and $\hat\delta$ indicate the expected number of induced aftershocks 
from a single earthquake of magnitude $m$. For example, an earthquake 
of magnitude 5, 5.5 and 6 are expected to induce $0.222$, $0.522$ and 
$1.229$ aftershocks under the estimated gamma \retas\ model. Under the 
Weibull \retas\ model, these values are slightly smaller at $0.179$, 
$0.421$, and $0.988$ respectively.  Therefore, we can conclude that
the choice of the renewal main-shock process is having a minimal
impact on the temporal and spatial distribution of aftershocks around
their triggering earthquake. However, even though the self-excited
part of the models are similar, the identification of main-shocks and
aftershocks will still be influenced by the main-shock arrival
rate which changes depending on the time lapsed since the most recent 
main-shock. 
%
%

For each model the temporal, longitudinal and latitudinal residual
series are computed and examined for uniformity and independence using
the K-S and L-B tests respectively. The p-values are reported in
Table~\ref{tab:NZ-GOF}. The classical
spatiotemporal \etas\ model does not provide a satisfactory fit since
the temporal residuals result in a rejection of the uniformity
assumption as the 5\% level with a p-value of only 0.023. Therefore,
it is unable to describe the temporal pattern of the earthquake arrival
times and additional flexibility in specifying the arrival time
distribution between earthquakes is required.
\begin{table}
	\caption{ \label{tab:NZ-GOF} GOF test results for the New Zealand earthquake catalog
		exponential (\etas), Weibull and gamma \retas\ models. This table 
		contains the p-values for the K-S tests and
		L-B tests for the temporal, longitudinal, and latitudinal
		residuals, as well as the combined residual series. 	 }
	
	\centering
		\begin{tabular}{l cc cc cc cc}
			\hline
			& \multicolumn{2}{c}{U} &\multicolumn{2}{c}{V} &\multicolumn{2}{c}{W} & \multicolumn{2}{c}{Combined} \\\cline{2-9}
			& K-S & L-B & K-S & L-B & K-S & L-B & K-S & L-B \\ 
			\hline
			\etas\ & 0.023 & 0.145 & 0.224 & 0.754 & 0.937 & 0.119 & 0.072 & 0.251\\ 
			{\retas}$_{Wei}$ & 0.159 & 0.058 & 0.120 & 0.693 & 0.902 & 0.104 & 0.135 & 0.117\\ 
			{\retas}$_{Gam}$ & 0.360 & 0.035 & 0.188 & 0.715 & 0.907 & 0.113 & 0.255 & 0.107\\
			\hline
	\end{tabular}
\end{table}

The second observation is that the p-values from the K-S test for the temporal residual series of the \retas\ model are
notably higher and signify a meaningful improvement in fit for both
the Weibull and gamma \retas\ model. This shows that the \retas\ model
can sufficiently model the temporal variation of earthquakes for this
New Zealand catalog. However, the p-value from the L-B test falls just short of the 
5\% signficance level for the gamma \retas\ model. For the spatial components, both the
longitudinal and latitudinal residuals indicate that all three models
are providing similar fits with p-values greater than 5\% for both
uniformity and independence. However, from the AIC values reported in
Table~\ref{tab:NZ-fits}, the two \retas\ models have markedly smaller
AIC values, with that of the gamma \retas\ model slightly smaller than
the Weibull \retas\ model. Therefore, by combining the GOF test
results and the AIC values, we conclude that the gamma \retas\ model
is the favored model to describe the seismicity for this New Zealand
earthquake catalog. 




\subsection{Earthquake catalogs around the world}
\label{sec:fits-world}

In this subsection, we study earthquake catalogs from various regions
around the world to confirm the utility of the \retas\ model in
modeling diverse range of seismically active regions. The section
includes an analysis of the following earthquake catalogs:
\begin{itemize}
	\item \textbf{Japan:} This catalog contains 577 earthquakes in the
	rectangular region defined by $141^\circ-145^\circ$ E and
	$36^\circ - 42^\circ$ N during the period 1980-01-01 until
	2015-06-30 with threshold magnitude $m_0 = 5.5$. This earthquake
	catalog was obtained from the Japan Meteorological Agency (JMA)
	database.
	
	\item \textbf{Chile:} This catalog contains 370 earthquakes in the
	rectangular region defined by $76^\circ-64^\circ$~W and
	$40^\circ-18^\circ$~S from 1997-01-01 until 2015-06-30 with
	threshold magnitude $m_0 = 5.5$. The data was sourced from the
	United States Geological Survey (USGS) database.
	
	\item \textbf{China:} This catalog contains 627 earthquakes in the
	region $97^\circ-107^\circ$~E and $26^\circ-34^\circ$~N, from
	1997-01-01 to 2015-06-30 with a threshold magnitude of $m_0 =
	4.5$. This data was also obtained from the USGS database.
\end{itemize}

For each of the three earthquake catalogs, we fit all three \retas\
models and summarize the results in Table~\ref{tab:world-fit}. For all
three earthquake catalogs the AIC value is significantly smaller for
the \retas\ model than the classical \etas\ model with the gamma 
\retas\ model performing the best based on this
criterion. The estimated shape parameter of the gamma
renewals $\alpha$ is statistically less than one in all data sets
, suggesting that the earthquakes in these catalogs are more heavily
clustered than implied by the \etas\ model. 
By examining the estimation results for the Japan earthquake catalog
we observe that there is a noticeable difference in the estimated
value $\hat p$ for each of the three models. For instance, in the
classical \etas\ model, $\hat p =1.004$ which is the smallest
estimated value, while for the gamma \retas\ model $\hat p =1.082$ is
significantly higher. This implies that for the \retas\ models, the
induced shocks happen much quicker than in the 
classical \etas\ model. 

\begin{table}
	\caption{	\label{tab:world-fit} Estimation results including parameter estimates and
		standard errors in parentheses for the Japan, Chile and China
		earthquake catalog with exponential (\etas), Weibull and gamma \retas\ models.}
	\centering
	\resizebox{\linewidth}{!}{
		\begin{tabular}{lccccccccc}
			\hline
			& $\hat \alpha$ & $\hat \beta$ & $\hat p$ & $\hat c$ & $\hat \sigma_1^2$ & $\hat \sigma_2^2$ & $\hat A$ & $\hat \delta$ & AIC  \\ 
			\hline
			\textbf{Japan} &&&&&&&&&\\
			\hline                      
			\multirow{2}{*}{\etas} & 1 & 87.96 & 1.004 & 0.0024 & 0.084 & 0.024 & 9.864 & 0.893 & \multirow{2}{*}{3947.68} \\ 
			& (NA)& (8.54) & (0.00044) & (0.00041) & (0.019) & (0.0033) & (1.451) & (0.100) & \\  \hline
			\multirow{2}{*}{{\retas}$_{Wei}$}  & 0.887 & 70.34 & 1.101 & 0.0059 & 0.084 & 0.024 & 0.663 & 0.814 & \multirow{2}{*}{3925.12} \\ 
			& (0.084) & (8.55) & (0.021) & (0.0016) & (0.018) & (0.0032) & (0.097) & (0.112)   & \\ \hline
			\multirow{2}{*}{{\retas}$_{Gam}$}  & 0.562 & 134.56 & 1.082 & 0.0052 & 0.075 & 0.023 & 0.727 & 0.829 & \multirow{2}{*}{3911.95} \\ 
			& (0.068) & (22.75) & (0.021) & (0.0014) & (0.013) & (0.0029) & (0.128) & (0.108)  & \\
			\hline
			\textbf{Chile} &&&&&&&&&\\ 
			\hline
			\multirow{2}{*}{\etas} & 1 & 39.63 & 1.043 & 0.0079 & 0.039 & 0.122 & 0.840 & 0.854 & \multirow{2}{*}{4320.26} \\ 
			& (NA)& (3.61) & (0.031) & (0.0032) & (0.0061) & (0.027) & (0.486) & (0.135)  & \\ \hline
			\multirow{2}{*}{{\retas}$_{Wei}$} & 0.708 & 32.37 & 1.018 & 0.0068 & 0.037 & 0.056 & 1.779 & 0.833 & \multirow{2}{*}{4303.60} \\ 
			& (0.052) & (3.91) & (0.015) & (0.0025) & (0.0051) & (0.012) & (1.406) & (0.127) & \\ \hline
			\multirow{2}{*}{{\retas}$_{Gam}$} & 0.565 & 69.68 & 1.016 & 0.0068 & 0.036 & 0.053 & 1.951 & 0.853 &  \multirow{2}{*}{4296.86} \\ 
			& (0.060) & (10.98) & (0.012) & (0.0024) & (0.0051) & (0.012) & (1.436) & (0.125) & \\ \hline
			\textbf{China} &&&&&&&&&\\
			\hline    
			\multirow{2}{*}{\etas} & 1& 43.02 & 1.103 & 0.016 & 0.044 & 0.0094 & 0.688 & 1.152 &  \multirow{2}{*}{3238.89} \\ 
			& (NA)&(4.02) & (0.024) & (0.0045) & (0.0036) & (0.0012) & (0.098) & (0.085) & \\  \hline
			\multirow{2}{*}{{\retas}$_{Wei}$}  & 0.818 & 36.90 & 1.118 & 0.018 & 0.044 & 0.0090 & 0.613 & 1.184 & \multirow{2}{*}{3232.49} \\ 
			& (0.059) & (4.24) & (0.022) & (0.0048) & (0.0036) & (0.0013) & (0.069) & (0.082) & \\ \hline
			\multirow{2}{*}{{\retas}$_{Gam}$}  & 0.694 & 60.20 & 1.100 & 0.016 & 0.044 & 0.0087 & 0.668 & 1.207 & \multirow{2}{*}{3229.45} \\ 
			& (0.081) & (9.21) & (0.024) & (0.0046) & (0.0036) & (0.0012) & (0.098) & (0.078) & \\ \hline
	\end{tabular}}
\end{table}

The estimated spatial response function which describes the
distribution of aftershocks around their triggering earthquake is very
similar in each the three models and with each catalog, and this is consistent with the New Zealand catalog presented in the previous section except for the longitudinal dispersion parameter for the Chile catalog which has an estimated value that is more than double that of the \retas\ models.
However, unlike the fitted models for the New Zealand
catalog, there are significant variations in the estimated values of the
parameter $A$ and $\delta$ between the three models. For instance, for
the Japan dataset, the estimated parameters for the classical \etas\
model are $\hat A = 9.864$ and $\hat \delta = 0.893$ while for the
gamma \retas\ model they are both considerably smaller with
$\hat A = 0.727$ and $\hat \delta = 0.829$. This implies that the 
classical \etas\ model has a larger number of
expected aftershocks induced from smaller magnitude earthquakes, and
since $\delta$ is also bigger it implies that more aftershocks on
average are induced from earthquakes of any magnitude. 

\section{Discussion}
\label{sec:conclusion}

This article proposed an improved version of the classical \etas\ model, in which the homogenous main-shock arrival process is replaced by inhomogeneous main-shock arrival process in the form of a renewal process. The {\retas} model was applied to several earthquake catalogs from around the world including New Zealand, Japan, Chile, and China. These earthquake catalogs revealed that the \retas\ model had a significant influence on the data fitting of main-shocks  compared to the classical \etas\ model. In fact, the classical \etas\ model was shown to not be suitable for the New Zealand earthquake catalog due to its simple assumption about main-shock arrival times. The \retas\ model overcame this shortfall and provided a superior quality of fit, as indicated by the newly proposed sequential GOF test procedure and the AIC value. After incorporating both Weibull and gamma renewal distributions, the estimated shape parameters $\hat \alpha$ of the main-shock renewal distributions (which is restricted to one in the classical \etas\ model) become significantly smaller than one, around 0.5-0.8, indicating a much stronger clustering of main-shock earthquakes.

We also fit the \retas\ model to two other seismically active regions
around the world, in California and Italy for the period 1997-01-01 to
2015-06-30 with a threshold magnitude $m_0 =4.5$ and $m_0 =4$,
respectively. For the California earthquake catagloue, the estimated
shape parameter for the gamma main-shock renewal distribution was
$\hat \alpha = 0.795$, which is again statistically significantly less
than one. The AIC value was $2370.71$, which is slightly smaller than
the classical \etas\ model at $2373.42$. However, the parameter
controlling the aftershock decay was estimated to be very close to one
using the density function $g$. Therefore, to allow $p$ to have a
values smaller than one, we use the temporal response function
$g(t) = (1+\frac t c)^{-p}$ and found $\hat p = 0.964$. However, this
means that $g$ can no longer be normalized to a properly density
function since its integral on $(0,\infty)$ diverges. The approximate 
Hessian matrix was also not invertible and therefore estimates of the 
standard errors were easily obtainable.

For the Italian earthquake catalog, the estimated shape parameter of
the gamma \retas\ model was $\hat \alpha = 1.041$, which is slightly
bigger than one. This indicates that the renewal main-shock arrival
process is not too dissimilar to a homogeneous Poisson process, and
the classical \etas\ model would be appropriate model for this
catalog. This is another advantage of the \retas\ model since it can
help in determining any departure or non-departure from a homogenous
main-shock arrival process. This is further reinforced by the AIC
values, which determines that the classical \etas\ model is the
superior model of choice when accounting for both the adequacy of
model fit and model complexity (number of parameters).

A more meaningful analysis of an earthquake catalog using the \retas\
model could be performed by utilizing a probability based stochastic
declustering algorithm, in which each earthquake is assigned a
probability to be either a main-shock or an aftershock induced by a
previous earthquake. For the classical \etas\ model, this is
straightforward to perform since the past points of the process are
conditionally independent of future points conditioned on the history. 
However, this is not exactly true for the \retas\ model, and therefore 
the probabilities required for the
declustering algorithm must be conditioned on the complete earthquake
catalog rather than only the past points. This provides an avenue for
future research.

\bibliographystyle{apalike}
\bibliography{SpatTemp}

\begin{thebibliography}{}

\bibitem[Chen and Hall, 2013]{Chen2013}
Chen, F. and Hall, P. (2013).
\newblock Inference for a nonstationary self-exciting point process with an
  application in ultra-high frequency financial data modeling.
\newblock {\em J. Appl. Probab.}, 50(4):1006--1024.

\bibitem[Chen and Stindl, 2018]{Chen2018}
Chen, F. and Stindl, T. (2018).
\newblock Direct likelihood evaluation for the renewal {H}awkes process.
\newblock {\em Journal of Computational and Graphical Statistics},
  27(1):119--131.

\bibitem[Console and Murru, 2001]{Console2001}
Console, R. and Murru, M. (2001).
\newblock A simple and testable model for earthquake clustering.
\newblock {\em Journal of Geophysical Research: Solid Earth},
  106(B5):8699--8711.

\bibitem[Console et~al., 2003]{Console2003}
Console, R., Murru, M., and Lombardi, A.~M. (2003).
\newblock Refining earthquake clustering models.
\newblock {\em Journal of Geophysical Research: Solid Earth}, 108(B10).

\bibitem[Console et~al., 2006]{Console2006}
Console, R., Rhoades, D.~A., Murru, M., Evison, F.~F., Papadimitriou, E.~E.,
  and Karakostas, V.~G. (2006).
\newblock Comparative performance of time-invariant, long-range and short-range
  forecasting models on the earthquake catalogue of {G}reece.
\newblock {\em Journal of Geophysical Research: Solid Earth}, 111(B9).

\bibitem[Daley and Vere-Jones, 2003]{Daley2003}
Daley, D.~J. and Vere-Jones, D. (2003).
\newblock {\em An Introduction to the Theory of Point Processes Volume I:
  Elementary Theory and Methods}.
\newblock Springer-Verlag, New York, 2nd edition.

\bibitem[{Godoy} et~al., 2016]{Godoy2016}
{Godoy}, B.~I., {Solo}, V., {Min}, J., and {Pasha}, S.~A. (2016).
\newblock Local likelihood estimation of time-variant {H}awkes models.
\newblock In {\em 2016 IEEE International Conference on Acoustics, Speech and
  Signal Processing (ICASSP)}, pages 4199--4203.

\bibitem[Gutenberg and Richter, 1944]{Gutenberg1944}
Gutenberg, B. and Richter, C.~F. (1944).
\newblock {Frequency of earthquakes in {C}alifornia}.
\newblock {\em Bulletin of the Seismological Society of America},
  34(4):185--188.

\bibitem[Harte, 2013]{Harte2013}
Harte, D.~S. (2013).
\newblock Bias in fitting the {ETAS} model: a case study based on {N}ew
  {Z}ealand seismicity.
\newblock {\em Geophysical Journal International}, 192(1):390--412.

\bibitem[Harte, 2014]{Harte2014}
Harte, D.~S. (2014).
\newblock An {ETAS} model with varying productivity rates.
\newblock {\em Geophysical Journal International}, 198(1):270--284.

\bibitem[Hawkes and Oakes, 1974]{Hawkes1974}
Hawkes, A.~G. and Oakes, D. (1974).
\newblock A cluster process representation of a self-exciting process.
\newblock {\em Journal of Applied Probability}, 11(3):493--503.

\bibitem[Helmstetter et~al., 2006]{Helmstetter2006}
Helmstetter, A., Kagan, Y.~Y., and Jackson, D.~D. (2006).
\newblock {Comparison of Short-Term and Time-Independent Earthquake Forecast
  Models for Southern California}.
\newblock {\em Bulletin of the Seismological Society of America},
  96(1):90--106.

\bibitem[Kolev and Ross, 2018]{Kolev2018}
Kolev, A.~A. and Ross, G.~J. (2018).
\newblock Inference for {ETAS} models with non-poissonian mainshock arrival
  times.
\newblock {\em Statistics and Computing}.

\bibitem[Kumazawa and Ogata, 2014]{Kumazawa2014}
Kumazawa, T. and Ogata, Y. (2014).
\newblock Nonstationary {ETAS} models for nonstandard earthquakes.
\newblock {\em Ann. Appl. Stat.}, 8(3):1825--1852.

\bibitem[Marzocchi and Lombardi, 2009]{Marzocchi2009}
Marzocchi, W. and Lombardi, A.~M. (2009).
\newblock Real-time forecasting following a damaging earthquake.
\newblock {\em Geophysical Research Letters}, 36(21).

\bibitem[Ogata, 1988]{Ogata1988}
Ogata, Y. (1988).
\newblock Statistical models for earthquake occurrences and residual analysis
  for point processes.
\newblock {\em Journal of the American Statistical Association}, 83(401):9--27.

\bibitem[Ogata, 1998]{Ogata1998}
Ogata, Y. (1998).
\newblock Space-time point-process models for earthquake occurrences.
\newblock {\em Annals of the Institute of Statistical Mathematics},
  50(2):379--402.

\bibitem[Ogata, 2004]{Ogata2004}
Ogata, Y. (2004).
\newblock Space-time model for regional seismicity and detection of crustal
  stress changes.
\newblock {\em Journal of Geophysical Research: Solid Earth}, 109(B3).

\bibitem[Omori, 1894]{Omori1894}
Omori, F. (1894).
\newblock On the aftershocks of earthquakes.
\newblock {\em Jounal of the College of Science, Imperial University of Tokyo},
  7:111--120.

\bibitem[Rosenblatt, 1952]{Rosenblatt1952}
Rosenblatt, M. (1952).
\newblock Remarks on a multivariate transformation.
\newblock {\em The Annals of Mathematical Statistics}, 23(3):470--472.

\bibitem[Schoenberg, 2003]{Schoenberg2003}
Schoenberg, F.~P. (2003).
\newblock Multidimensional residual analysis of point process models for
  earthquake occurrences.
\newblock {\em Journal of the American Statistical Association},
  98(464):789--795.

\bibitem[Utsu, 1961]{Utsu1961}
Utsu, T. (1961).
\newblock A statistical study of the occurrence of aftershocks.
\newblock {\em Geophysical Magazine}, 30:521--605.

\bibitem[Werner et~al., 2011]{Werner2011}
Werner, M.~J., Helmstetter, A., Jackson, D.~D., and Kagan, Y.~Y. (2011).
\newblock {High-Resolution Long-Term and Short-Term Earthquake Forecasts for
  California}.
\newblock {\em Bulletin of the Seismological Society of America},
  101(4):1630--1648.

\bibitem[Wheatley et~al., 2016]{Wheatley2016}
Wheatley, S., Filimonov, V., and Sornette, D. (2016).
\newblock The {H}awkes process with renewal immigration $\&$ its estimation
  with an {EM} algorithm.
\newblock {\em Computational Statistics $\&$ Data Analysis}, 94:120 -- 135.

\bibitem[Zhuang et~al., 2005]{Zhuang2005}
Zhuang, J., Chang, C.-P., Ogata, Y., and Chen, Y.-I. (2005).
\newblock A study on the background and clustering seismicity in the {T}aiwan
  region by using point process models.
\newblock {\em Journal of Geophysical Research: Solid Earth}, 110(B5).

\bibitem[Zhuang et~al., 2008]{Zhuang2008}
Zhuang, J., Christophersen, A., Savage, M.~K., Vere-Jones, D., Ogata, Y., and
  Jackson, D.~D. (2008).
\newblock Differences between spontaneous and triggered earthquakes: Their
  influences on foreshock probabilities.
\newblock {\em Journal of Geophysical Research: Solid Earth}, 113(B11).

\bibitem[Zhuang et~al., 2002]{Zhuang2002}
Zhuang, J., Ogata, Y., and Vere-Jones, D. (2002).
\newblock Stochastic declustering of space-time earthquake occurrences.
\newblock {\em Journal of the American Statistical Association},
  97(458):369--380.

\bibitem[Zhuang et~al., 2004]{Zhuang2004}
Zhuang, J., Ogata, Y., and Vere-Jones, D. (2004).
\newblock Analyzing earthquake clustering features by using stochastic
  reconstruction.
\newblock {\em Journal of Geophysical Research: Solid Earth}, 109(B5).

\end{thebibliography}

\end{document}